\documentclass[a4paper]{article}
\usepackage{a4wide}
\usepackage[normalem]{ulem}
\usepackage{dsfont}
\usepackage[T1]{fontenc}
\usepackage{amsmath}
\usepackage{amssymb}
\usepackage{amsthm}
\usepackage[ansinew]{inputenc}
\usepackage{enumerate}
\usepackage{fancyhdr}
\usepackage{color}
\usepackage{ifpdf}
\usepackage{graphicx}
\usepackage{rotating}
\usepackage[T2A]{fontenc}
\usepackage{mathtext}
\usepackage{caption}
\newtheorem{thm}{Theorem}
\newtheorem{prop}{Proposition}
\newtheorem{cor}{Corollary}

\newtheorem*{prop*}{Proposition}
\newcommand{\manifold}{\mathcal{M}}
\newcommand{\massshell}{\mathcal{H}}
\usepackage{authblk}

\begin{document}
 \title{Self-similarity breaking of cosmological solutions with collisionless matter}

\author[1]{Ho Lee\footnote{holee@khu.ac.kr}}
\author[2]{Ernesto Nungesser\footnote{ernesto.nungesser@icmat.es}}

\affil[1]{Department of Mathematics and Research Institute for Basic Science, Kyung Hee University, Seoul, 02447, Republic of Korea}
\affil[2]{Instituto de Ciencias Matem\'{a}ticas (CSIC-UAM-UC3M-UCM), 28049 Madrid, Spain}
 
\maketitle
\begin{abstract}
In this paper we consider the Einstein-Vlasov system with Bianchi VII$_0$ symmetry. Under the assumption of small data we show that self-similarity breaking occurs for reflection symmetric solutions. This generalizes the previous work concerning the non-tilted fluid case \cite{WHU} to the Vlasov case, and we obtain detailed information about the late-time behaviour of metric and matter terms.
 \end{abstract}

\section{Introduction}

Due to the observations of the microwave background we know that there is electromagnetic radiation all over the place which is almost isotropic. Only small temperature variations of $10^{-5} K$ are present, which makes them slightly anisotropic. The question is whether this implies that the spacetime has to be almost isotropic as well. If the matter distribution is exactly isotropic, there is a known result, called the EGS-theorem \cite{EGS}, which tells you that the spacetime will also be isotropic if the collision term vanishes, i.e.\ in the Vlasov case or if there is a so-called detailed balance. These results were extended to the Boltzmann case in \cite{TE}. There are also ``Almost'' EGS-theorems, however a critical question is what almost means. We refer for an overview of these results to \cite{NUWL} and Section 4.1 of \cite{Hans}.

On the other hand there is a very interesting result by Wainwright, Hancock and Uggla \cite{WHU} which points in a negative direction at least if there is no cosmological constant present and in a homogeneous setting. There it is shown for the case of a perfect fluid as a matter model that a solution of Bianchi type VII$_0$ can be arbitrary close to a FLRW solution, but still can become very anisotropic as regards the Weyl curvature. Here we show a similar result for collisionless matter with a small momenta dispersion. In our case it is a small data result, where one curvature variable $N_+$ is very big initially, which means that its inverse $M$, is very small. The other curvature variable $N_-$ and the shear variables $\Sigma_+$ and $\Sigma_-$ are small initially as well. What is obtained is an analogue to the result of \cite{WHU} using the methods of \cite{EN}. 

Note that the Weyl curvature becomes unbounded, but there are no upper bounds to the Weyl curvature from observations. The Weyl curvature might be of importance from theoretical considerations concerning the initial singularity or even in a broader cosmological context \cite{CCC}. Another aspect which makes this Bianchi model interesting is that there are oscillations in both $N_-$ and $\Sigma_-$. In order to do the analysis we follow the analysis of \cite{WHU} which consists in separating oscillating and non-oscillating part.

Finally let us note that in the perfect fluid analysis, the dust case is particularly important because a Weyl curvature bifurcation occurs for the value $\gamma=1$ which corresponds to dust. If the value of $\gamma$ is smaller then one, then the Weyl curvature will tend to zero. For dust the Weyl curvature tends to a constant, and if it is bigger than one the Weyl curvature becomes unbounded cf. Theorem 2.4 of \cite{WHU}.

\section{The Einstein-Vlasov system}
In this section we introduce the Einstein-Vlasov system with Bianchi symmetry. Consider a four-dimensional oriented and time oriented Lorentzian manifold $(\manifold, {^4g})$ and a distribution function $f$, then the Einstein-Vlasov system is written as
\begin{align*}
G_{\alpha\beta}&= T_{\alpha \beta},\\
\mathcal{L} f&=0,
\end{align*}
where $G_{\alpha\beta}$ is the Einstein tensor and $T_{\alpha\beta}$ is the energy-momentum tensor defined by
\begin{align*}
T_{\alpha\beta}=  \int_{\massshell} \chi p_{\alpha} p_{\beta}.
\end{align*}
Here, the integration is over the mass-shell $\massshell$ at a given space-time point which is defined by
\[
p_{\alpha} p_{\beta}g^{\alpha\beta}=-1
\]
for massive particles, and $\chi$ is the distribution function multiplied by the Lorentz invariant measure and $\mathcal{L}$ the Liouville operator.

The basic equations we will use can be found in Sections 7.3--7.4 and Chapter 25 of \cite{Hans}. We also refer to this book for an introduction to the Einstein-Vlasov system. Let $\Sigma$ be a spacelike hypersurface in $\manifold$ with $n$ its future directed unit normal. We define the second fundamental form as $k(X,Y)=\langle \nabla_X n, Y\rangle$ for vectors $X$ and $Y$ tangent to $\Sigma$, where $\nabla$ is the Levi-Civita connection of $^4g$. The Hamiltonian and momentum constraint equations are as follows:
\begin{align*}
&R-{k}_{ij} {k}^{ij}+ {k}^2=2 \rho,\\
&\overline{\nabla}^j {k}_{ji}-\overline{\nabla}_i  {k}= -{J}_i,
\end{align*}
where ${g}$ is the induced metric on $\Sigma$, $k=k_{ab}g^{ab}$ the trace of the second fundamental form, $R$ and $\overline{\nabla}$ the scalar curvature and the Levi-Civita connection of ${g}$ respectively, and matter terms are given by $\rho= T_{\alpha\beta}n^{\alpha}n^{\beta}$ and $J_i X^i=- T_{\alpha\beta}n^{\alpha}X^{\beta}$ for $X$ tangent to $\Sigma$. Here and throughout the paper we assume that Greek letters run from $0$ to $3$, while Latin letters from $1$ to $3$, and also follow the sign conventions of \cite{Hans}.

\subsection{The Einstein-Vlasov system with Bianchi symmetry}

A Bianchi spacetime is defined to be a spatially homogeneous spacetime whose isometry group possesses a three-dimensional subgroup that acts simply transitively on spacelike orbits. A Bianchi spacetime admits a Lie algebra of Killing vector fields. These vector fields are tangent to the group orbits, which are the surfaces of homogeneity. Using a left-invariant frame, the metric induced on the spacelike hypersurfaces depends only on the time variable. Let $G$ be the three-dimensional Lie group, $e_i$ a basis of the Lie algebra, and $\xi^i$ the dual of $e_i$. The metric of the Bianchi spacetime in the left-invariant frame is written as
\begin{align*}
^4 g =-dt \otimes dt + g_{ij}\xi^i \otimes \xi^j
\end{align*}
on $\manifold=I \times G$ with $e_0$ future oriented. We will need equations (25.17)--(25.18) of \cite{Hans} (without scalar field) with the notation $T_{ab}=S_{ab}$:
\begin{align}
&\dot{g}_{ab}=2k_{ab},\label{a} \\
&\dot{k}_{ab}=-R_{ab}+2 k^i_a k_{bi} -k\, k_{ab}+S_{ab}+\frac12 (\rho-S) g_{ab}\label{EE},
\end{align}
where the dot means the derivative with respect to time $t$ and $S=g^{ab} S_{ab}$ and $R_{ab}$ is the Ricci tensor associated to the induced 3-metric. Since $k$ does not depend on spatial variables, the constraint equations are as follows: 
\begin{align}
&R-{k}_{ij} {k}^{ij}+ {k}^2=2 \rho,\label{CE1} \\
&\nabla^j {k}_{ji}= -{J}_i,\label{CE2}
\end{align}
where we have dropped the bar on the covariant derivative by a slight abuse of notation.

Below, we collect and derive several useful equations. The mixed version of the second fundamental form is given by
\begin{align}\label{MV}
\dot{k}^a_b= -R^a_b-k\,k^a_b +S^a_b + \frac12(\rho-S) \delta^a_b,
\end{align}
and by taking the trace of (\ref{MV}) we have
\begin{align}\label{im}
\dot k=-R-k^2-\frac12  S +\frac32 \rho.
\end{align}
By the constraint equation \eqref{CE1} one can eliminate the energy density  such that (\ref{im}) reads:
\begin{align}\label{in}
\dot k=-\frac{1}{4}(k^2+R+3k_{ab}k^{ab})-\frac12 S,
\end{align}
and if we substitute for the Ricci scalar with (\ref{CE1}), we obtain
\begin{align}\label{loo}
\dot k=-k_{ab}k^{ab}-\frac12 (S+\rho).
\end{align}
It is convenient to express the second fundamental form as
\begin{align*}
k_{ab}=\sigma_{ab}+H g_{ab},
\end{align*}
where $\sigma_{ab}$ is the trace free part and $H$ is the Hubble parameter:
\begin{align*}
H=\frac{1}{3}k,
\end{align*}
and we obtain
\begin{align*}
k_{ab}k^{ab}=\sigma_{ab}\sigma^{ab}+3H^2.
\end{align*}
By a simple calculation we obtain from (\ref{in})  
\begin{align}\label{H-1}
\partial_t(H^{-1})=\frac{3}{2}+\frac{R}{12H^2}+ \frac{\Sigma_a^b\Sigma^a_b}{4}+\frac{S}{6H^2},
\end{align}
where we have defined
\begin{align*}
\Sigma_a^b=\frac{\sigma_a^b}{H}.
\end{align*}
Moreover, let us introduce the following quantities:
\begin{align*}
\Omega=\frac{\rho}{3H^2}, \quad q=-1-\frac{\dot{H}}{H^2}.
\end{align*}
We can see that the constraint equation (\ref{CE1}) is written as
\begin{align}\label{constraint}
\Omega+ \frac{1}{6}\bigg(\Sigma^i_j\Sigma^j_i-\frac{R}{H^2}\bigg)=1.
\end{align}
Except for Bianchi IX, the scalar curvature is always non-positive (see the expression (E.12) in Appendix E.2 on page 699 of \cite{Hans}) so that
\begin{align*}
\Omega \leq 1.
\end{align*}
We also derive an evolution equation of $\sigma_{ab}$ by combining the above results:
\begin{align*}
\dot{\sigma}_{a}^{b}=\dot{k}^b_a-\dot{H}\delta^{b}_a= -3H \sigma_{a}^b+S_{a}^b-\frac13 S \delta_a^b-R^b_a+\frac13 R\delta^b_a.
\end{align*}
We now define the shear variables:
\begin{align*}
\Sigma_{+}=\frac12(\Sigma_{2}^2+\Sigma_{3}^3)=-\frac12 \Sigma^1_1, \quad  \Sigma_{-}=\frac{1}{2\sqrt{3}}(\Sigma_{2}^2-\Sigma_{3}^3),
\end{align*}
and the dimensionless time variable $\tau$:
\begin{align*}
\frac{dt}{d\tau}=H^{-1},
\end{align*}
and combine (\ref{MV}), (\ref{H-1}), and (\ref{constraint}) to obtain the evolution equations for $\Sigma_-$ and $\Sigma_+$:
\begin{align}
\label{Bla1}&{\Sigma}_+'=(q-2)\Sigma_++\frac{2R-3(R^2_2+R^3_3)}{6H^2}+S_+,\\
\label{Bla2}&{\Sigma}_-'=(q-2)\Sigma_-+\frac{R_3^3-R^2_2}{2\sqrt{3}H^2}+S_-,
\end{align}
where we have used the notation
\begin{align*}
S_+&=\frac{1}{6H^2}(3S^2_2+3S^3_3-2S),\\
S_-&=\frac{S^2_2-S^3_3}{2\sqrt{3}H^2}.
\end{align*}

\subsection{Vlasov equation with Bianchi symmetry}
Since we use a left-invariant frame, $f$ will not depend on $x^a$. Moreover, since $g_{00}=g^{00}=-1$ and $g^{0a}=0$, we have $p^0=-p_0=\sqrt{1+p_ap_bg^{ab}}$, $\rho=T_{00}$, and $J_{a}=-T_{0a}$. The frame components of the energy-momentum tensor are thus
\begin{align*}
&\rho=(\det g)^{-\frac12}  \int f(t,p_*) \sqrt{1+p_ap_bg^{ab}} dp_*,\\
&J_{i}=(\det g)^{-\frac12} \int f(t,p_*) p_i  dp_*,\\
&S_{ij}=(\det g)^{-\frac12} \int f(t,p_*) \frac{p_i p_j} {\sqrt{1+p_ap_bg^{ab}}}dp_*,
\end{align*}
where the distribution function is understood as $f=f(t,p_*)$ with $p_*=(p_1,p_2,p_3)$.

We define ${P}$ as the supremum of the square of momenta at a given time:
\begin{align}\label{PP}
P(t)= \sup\{ g^{ab}p_ap_b \vert f(t,p_*)\neq 0 \}.
\end{align}

\begin{prop}
Consider the Vlasov equation in a Bianchi spacetime. Then, the support in momentum space is bounded as follows: 
\begin{align*}
{P}(\tau) \leq {P} (\tau_0) \exp \Big(2\int_{\tau_0}^\tau (-1+(\Sigma^a_b \Sigma^b_a)^{\frac12}) ds\Big).
\end{align*}
\end{prop}
\begin{proof}
The Vlasov equation is written as
\begin{align}\label{ve}
p^0\frac{\partial f}{\partial t}+C^d_{ba}p^{b}p_{d}\frac{\partial f}{\partial p_a}=0,
\end{align}
where $C^d_{ba}$ are the structure constants of the Lie algebra (see \cite{LN} for details). A characteristic curve is defined for each $V_a(t)=p_a$ by
\begin{align}\label{charak}
\frac{dV_a}{dt}=(V^0)^{-1}C^d_{ba}V^bV_{d}.
\end{align}
For the rest of the paper the capital $V_a$ will indicate that $p_a$ is parametrised by the coordinate time $t$ (or $\tau$ if we express all the variables in these terms). Note that if we define
\begin{align}\label{VVV}
 \tilde{P}=g^{ab}V_aV_b,
\end{align}
due to the antisymmetry of the structure constants we have from (\ref{charak})
\begin{align}\label{ha}
\dot{\tilde{P}}=\dot{g}^{ab}V_aV_b.
\end{align}
In the sense of quadratic forms we have
\begin{align*}
\sigma^{ab} \leq (\sigma^c_d \sigma^d_c)^{\frac12} g^{ab},
\end{align*}
and as a consequence
\begin{align}\label{decayk}
\dot{g}^{ab} \leq 2H (-1+(\Sigma^a_b \Sigma^b_a)^{\frac12}) g^{ab}.
\end{align}
We combine (\ref{ha}) and (\ref{decayk}) with the dimensionless time variable $\tau$ to obtain
\begin{align*}
\tilde{P}'\leq 2 (-1+(\Sigma^a_b \Sigma^b_a)^{\frac12})\tilde{P}.
\end{align*}
Integrating this,
\begin{align}\label{PPP}
\tilde{P}(\tau) \leq \tilde{P} (\tau_0) \exp \Big(2\int_{\tau_0}^\tau (-1+(\Sigma^a_b \Sigma^b_a)^{\frac12}) ds\Big),
\end{align}
and taking the supremum, we obtain the desired result.
\end{proof}
As a consequence we have a nice result in the case of massive particles in a Bianchi spacetime which is not Bianchi IX since the quotient of the trace of the energy-momentum tensor and the energy density is in fact bounded by the support of the momenta.
\begin{cor}
Consider massive solutions to the Einstein-Vlasov system in a Bianchi spacetime which is not Bianchi IX. If the shear is bounded like
$(\Sigma^a_b \Sigma^b_a)^{\frac12} < 1$ for all times, we have a dust-like behaviour in the sense that
\begin{align*}
\frac{S}{\rho} \rightarrow 0.
\end{align*}
\end{cor}
\begin{proof}
Denoting by a hat that we use an orthonormal frame we have
\begin{align*}
 S=S_{ab}g^{ab}=\hat{S}_{11}+\hat{S}_{22}+\hat{S}_{33},
 \end{align*} 
 so that the following holds for the trace 
\begin{align*}
S&=\int f(t,\hat{p})\frac{|\hat{p}|^2}{p^0}d\hat{p}\leq P\int f(t,\hat{p})|\hat{p}|d\hat{p}\leq P\rho,
\end{align*}
where we used $|\hat{p}|^2=\hat{p}_1^2+\hat{p}_2^2+\hat{p}_3^2=g^{ab}p_ap_b$ and $p^0\geq 1$. Applying the previous result we obtain
\[
\frac{S}{\rho}\leq P (\tau_0) \exp \Big(2\int_{\tau_0}^\tau (-1+(\Sigma^a_b \Sigma^b_a)^{\frac12}) ds\Big),
\]
which tends to zero as $\tau\to\infty$.
\end{proof}
For later use we remark that the quantity $S/(3H^2)$ has the same bound with $S/\rho$ as follows:
\begin{align}\label{SS}
\frac{S}{3H^2}  \leq \frac{\rho}{3H^2} P= \Omega P \leq P.
\end{align}
Thus we have that
\begin{align}\label{SSS}
\frac{S}{3H^2}  \leq P (\tau_0) \exp \Big(2\int_{\tau_0}^\tau (-1+(\Sigma^a_b \Sigma^b_a)^{\frac12}) ds\Big).
\end{align}

\subsection{Reflection and Bianchi VII$_0$ symmetries}
We will assume an additional symmetry namely the reflection symmetry:
\begin{align*}
 f(t,p_1,p_2,p_3)=f(t,-p_1,-p_2,p_3)=f(t,p_1,-p_2,-p_3).
\end{align*}
If we suppose that the distribution function is initially reflection symmetric and the metric and the second fundamental form are initially diagonal, then we see that the energy-momentum tensor is diagonal as well. One can see from the equations that the metric and the second fundamental form will remain diagonal. This symmetry implies in particular that there is no matter current. In the diagonal case there is a simple formula for the Ricci tensor. Let $(ijk)$ denote a cyclic permutation of $(123)$ and let us suspend the Einstein summation convention for the next three formulas. Introduce $\nu_i$ as the signs depending on the Bianchi type (see Table 1 of \cite{CH}), and define 
\begin{align*}
n_i=\nu_i \sqrt{\frac{g_{ii}}{g_{jj}g_{kk}}},
\end{align*}
which implies $n_2= \nu_2 \sqrt{g_{22}/(g_{33}g_{11})}$, etc, and the Ricci tensor (cf.\ (11a) of \cite{CH}) is given by
\begin{align*}
R^i_i=\frac{1}{2}\Big(n_i^2 -(n_j-n_k)^2\Big).
\end{align*}
We also have in the diagonal case that
\begin{align*}
\Sigma_+^2+\Sigma_-^2= \frac{1}{6}\Big((\Sigma^1_1)^2+(\Sigma_{2}^2)^2+(\Sigma_{3}^3)^2\Big)= \frac16\Sigma^a_b\Sigma^b_a.
\end{align*}

We will study the Bianchi VII$_0$ case. For Bianchi VII$_0$, since $\nu_1=0$ and $\nu_2=\nu_3=1$, we have that $n_{2}$ and $n_{3}$ are positive definite and the only non-vanishing structure constants are (cf.\ Appendix E, page 695 of \cite{Hans}):
\begin{align}\label{sc7}
C^2_{31}=1=-C^2_{13},\quad C^3_{12}=1=-C^3_{21}.
\end{align}
The curvature expressions are
\begin{align*}
& R_{1}^1=R=-\frac{1}{2} (n_2-n_3)^2,\\
&R^2_2=-R^3_3=\frac{1}{2}(n_2^2-n_3^2).
\end{align*}
Let $N_{ii}$ denote $n_i/H$, then straightforward computations lead to
\begin{align*}
N_{22}' = N_{22} (q+2\Sigma_++2\sqrt{3}\Sigma_-),\\
N_{33}' = N_{33} (q+2\Sigma_+-2\sqrt{3}\Sigma_-),
\end{align*}
(see \cite{E3} for more details in the Bianchi II case). In order to compare this with the equations of \cite{WHU} we use the following quantities
\begin{align*}
&N_+=\frac{N_{22}+N_{33}}{2}>0,\\
&N_-=\frac{N_{22}-N_{33}}{2\sqrt{3}}.
\end{align*}
The evolution equations for the shear and curvature variables using the dimensionless time variable are as follows:
\begin{align*}
&\Sigma_+'=(q-2)\Sigma_+-2N_-^2+S_+,\\
&\Sigma_-'=(q-2)\Sigma_--2N_+N_-+S_-,\\
&N_+'=(q+2\Sigma_+)N_++6\Sigma_-N_-,\\
&N_-'=(q+2\Sigma_+)N_-+2\Sigma_-N_+,
\end{align*}
which one can compare with (3.10) of \cite{WHU}. Since we expect a similar behaviour as in the fluid case, we introduce as (3.15)--(3.16) in \cite{WHU} the following variables (we use $X$ instead of $R$ to avoid confusion with the scalar curvature):
\begin{align*}
&M=\frac{1}{N_+}>0,\\
&N_-=X \sin \psi,\quad X>0,\\
&\Sigma_-=X \cos \psi, \quad X>0.
\end{align*}
We also split $q$ in what we expect to be the non-oscillatory part $Q$ and the oscillatory part
\begin{align*}
q&=-1-\frac{\dot{H}}{H^2}=\frac12+\frac{R}{12H^2}+\frac{\Sigma^i_j \Sigma^j_i}{4} + \frac{S}{6H^2}\\
&=\frac12-\frac12 N_-^2+\frac32 \Sigma_+^2 +\frac32 \Sigma_-^2+\frac{S}{6H^2}\\
&=\frac12 +\frac32 \Sigma_+^2 +\frac12 X^2+X^2 \cos 2\psi+\frac{S}{6H^2}=Q+X^2 \cos 2\psi,
\end{align*}
where 
\begin{align*}
Q=\frac12 +\frac32 \Sigma_+^2 +\frac12 X^2 +\frac{S}{6H^2}.
\end{align*}
The differential equations for $\Sigma_+$, $\Sigma_-$, $N_+$, and $N_-$ are now written as follows:
\begin{align}
&\label{+} \Sigma_+'=\Sigma_+(Q-2)-X^2+(1+\Sigma_+)X^2\cos 2\psi +S_+,\\
&\label{M} {M}'=-M[Q+2\Sigma_++X^2(\cos 2\psi+3M\sin 2\psi)],\\
&\label{X} X'=[Q+\Sigma_+-1+(X^2-1-\Sigma_+) \cos 2\psi ]X+\cos \psi S_- ,\\
&\label{psi} \psi'=\frac{1}{M}\Big(2+M[(1+\Sigma_+)\sin 2\psi - X^{-1}S_-\sin \psi]\Big),
\end{align}
which can be compared to (3.17)--(3-20) of \cite{WHU}. An important quantity in the analysis of \cite{WHU} is what they call the Weyl parameter. It is a dimensionless measure of the Weyl curvature tensor. We have defined shear and curvature variables in such a way that we have the same expressions as in \cite{WHU}. The square of it is defined as (cf.\ (3.37)--(3.38) of \cite{WHU})
\begin{align*}
\mathcal{W}^2= \mathcal{E_+}^2+  \mathcal{E_-}^2+ \mathcal{E_-}^2+\mathcal{H_-}^2,
\end{align*}
where 
\begin{align*}
& \mathcal{E_+}=\Sigma_+(1+\Sigma_+)+\frac12X^2(1-3\cos 2\psi),\\
&\mathcal{E_-}=\frac{2X}{M}\Big(\sin \psi +\frac12 M(1-2\Sigma_+)\cos \psi\Big),\\
&\mathcal{H_+}=-\frac32 X^2 \sin 2\psi,\\
&\mathcal{H_-}=\frac{2X}{M}\Big(-\cos \psi - \frac32 M \Sigma_+ \sin \psi\Big).
\end{align*}
For details we refer to \cite{WHU} and references therein.
\section{The bootstrap argument}
\subsection{Bootstrap assumptions}
We will use a bootstrap argument. Let us assume that there exists an interval $[\tau_0,\tau_1)$ where the following estimates hold:
\begin{align*}
M(\tau)  &\leq \varepsilon e^{-\frac25 (\tau-\tau_0)},\\
X(\tau) &\leq \varepsilon e^{-\frac25 (\tau-\tau_0)},\\
\vert \Sigma_+(\tau)\vert& \leq \varepsilon e^{-\frac35 (\tau-\tau_0)},
\end{align*}
where the epsilon is a small positive constant. 
\subsection{Matter terms}
Due to the bootstrap assumptions the shear is bounded. Since $\Sigma^a_b \Sigma^b_a = 6 (\Sigma^2_+ +\Sigma_-^2)$, we have
\begin{align*}
\Sigma^a_b \Sigma^b_a \leq 6(\varepsilon^2 + X^2\cos^2 \psi) \leq C\varepsilon^2.
\end{align*}
As a consequence using \eqref{SSS} we have
\begin{align*}
\frac{S}{3H^2} \leq P (\tau_0) \exp \Big(2 \int_{\tau_0}^\tau (-1+C\varepsilon)ds\Big)\leq P(\tau_0)e^{(-2+C\varepsilon)(\tau-\tau_0)}.
\end{align*}
The variables $S_+$ and $S_-$ are bounded by that up to an irrelevant constant. Note that
\begin{align*}
S_+=\frac{1}{6H^2}(S-3S^1_1),\quad S_-=\frac{1}{2\sqrt{3}H^2}(S^2_2-S_3^3),
\end{align*}
hence we have
\begin{align*}
(S_+)^2+(S_-)^2&=\frac{1}{36H^4}(S^2-6SS^1_1+9(S^1_1)^2)+\frac{1}{12H^4}((S_2^2)^2+(S^3_3)^2-2S^2_2S^3_3)\\
&=\frac{1}{36H^4}(S^2+3((S^1_1)^2+(S_2^2)^2+(S^3_3)^2)-6SS^1_1+6(S_1^1)^2-6S^2_2S^3_3)\\
&=\frac{1}{36H^4}(S^2+6((S^1_1)^2+(S_2^2)^2+(S^3_3)^2)-3(S^1_1+S^2_2+S^3_3)^2)\\
&=\frac{1}{6H^4}S^a_bS^b_a-\frac{1}{18H^4}S^2.
\end{align*}
Since $S^a_bS^b_a\leq S^2$, we have $(S_+)^2+(S_-)^2 \leq S^2/(9H^4)$, and conclude that $S_+$ and $S_-$ are bounded as follows:
\[
|S_\pm|\leq P(\tau_0)e^{(-2+C\varepsilon)(\tau-\tau_0)}.
\]

\subsection{Estimate of $M$}
In order to obtain a better estimate for $M$ it will be necessary to estimate the term in brackets of \eqref{M} which can be expressed using the definition of $Q$ as follows:
\begin{align}\label{uplow}
-\frac12-\frac{d}{d\tau}\log M= 2\Sigma_+ +\frac32 \Sigma_+^2+X^2\Big(\frac12 +\cos 2\psi+3M\sin 2\psi\Big)+\frac{S}{6H^2}.
\end{align}
Using the bootstrap assumptions and the bound for the matter term we obtain for the evolution equation of $M$
\begin{align*}
\Big|\frac12+\frac{d}{d\tau}\log M\Big|\leq C(\varepsilon+P(\tau_0)),
\end{align*}
which implies that
\begin{align}\label{estimatem}
M(\tau_0)  e^{(-\frac12-C(\varepsilon+P(\tau_0)))(\tau-\tau_0)}  \leq  M(\tau)  \leq  M(\tau_0)  e^{(-\frac12+C(\varepsilon+P(\tau_0)))(\tau-\tau_0)}.
\end{align}

\subsection{Estimate of $X$}
The evolution equations for $M$, $X$, and $\psi$ are written as
\begin{align*}
&M'=-dM ,\\
&X'=(a+b \cos 2 \psi) X +c,\\
&\psi'= 2M^{-1}  + e,
\end{align*}
where $a$, $b$, $c$, $d$, and $e$ are functions of $\tau$ given by
\begin{align}
&a=Q+\Sigma_+-1,\\
&b=X^2-1-\Sigma_+,\\
&c=S_- \cos \psi,\\
&d=Q+2\Sigma_++X^2(\cos 2\psi+3M\sin 2\psi),\\
&e=(1+\Sigma_+)\sin 2\psi - X^{-1}S_-\sin \psi.
\end{align}
An oscillation term appears in the evolution equation of $X$ with the factor $b$ which is not small. In order to obtain an estimate of $X$ we have to get rid of the $b$-factor. We introduce the variable
\begin{align*}
\bar{X}=f X\quad\mbox{with}\quad f=\frac{1}{1+\frac14 Mb \sin 2\psi}.
\end{align*}
Note that $f>0$ as long as $M$, $X$, and $\Sigma_+$ are small. By direct calculations we have
\begin{align*}
f'=f^2 M\bigg( \frac{bd-b'}{4}\sin 2\psi -\frac{be}{2}  \cos 2\psi\bigg)-f^2 b \cos 2\psi,
\end{align*}
and the evolution equation of $\bar{X}$ is given by
\begin{align}\label{barx}
\bar{X}'=\bigg( a+M\bigg(\frac{ab+bd-b'}{4}\sin 2\psi-\frac{be}{2}\cos2\psi+\frac{b^2}{8}\sin 4\psi\bigg)\bigg) f \bar{X}+ cf.
\end{align}
We need the following estimates. By the definition of $Q$ we have
\begin{align*}
a=\frac12+\frac32\Sigma_+^2+\frac12 X^2+\frac{S}{6H^2}+\Sigma_+-1,
\end{align*}
which implies that $|a+\frac12|\leq C(\varepsilon+P(\tau_0))$ by the bootstrap assumption. Similarly we have
\[
|b+1|\leq C\varepsilon,\quad \bigg|d-\frac12\bigg|\leq C(\varepsilon+P(\tau_0)).
\]
For $b'$ we write
\begin{align*}
b'&=2cX+X^2(2a+2b\cos 2\psi -1 +\cos 2\psi)\\
&\quad +\frac32 \Sigma_+-\Sigma_+\bigg(-\frac32+\frac32\Sigma_+^2+\frac12 X^2-X^2\cos 2\psi +\frac{S}{6H^2}\bigg) +S_+,
\end{align*}
and this shows that $|b'|\leq C(\varepsilon+P(\tau_0))$ by the bootstrap assumptions and the estimates above. Consequently, we obtain 
\begin{align*}
\bar{X}'&=\bigg( a+M\bigg(\frac{ab+bd-b'}{4}\sin 2\psi-\frac{be}{2}\cos2\psi+\frac{b^2}{8}\sin 4\psi\bigg)\bigg) f \bar{X}+ cf\\
&=\bigg( a+M\bigg(\frac{ab+bd-b'}{4}\sin 2\psi-\frac{b(1+\Sigma_+)\sin 2\psi}{2}\cos2\psi+\frac{b^2}{8}\sin 4\psi\bigg)\bigg) f \bar{X}\\
&\quad +\frac{Mb}{2}S_-\sin\psi\cos 2\psi f^2 +cf.
\end{align*}
Applying the estimates of $a$, $b$, $d$, and $b'$ together with the bootstrap assumptions and the estimates of matter terms, we obtain
\begin{align*}
\bar{X}'&\leq \bigg(-\frac12 +C(\varepsilon +P(\tau_0))\bigg)f\bar{X}+(C\varepsilon f^2 +f)S_-\\
&\leq \bigg(-\frac12 +C(\varepsilon +P(\tau_0))\bigg)\bar{X}+(1+C\varepsilon )P(\tau_0)e^{(-2+C\varepsilon)(\tau-\tau_0)},
\end{align*}
where we used $f= 1+O(\varepsilon)$ by the bootstrap assumption, and this can be written as
\begin{align*}
\frac{d}{d\tau}\Big[e^{\frac12(\tau-\tau_0)}\bar{X}\Big] &\leq C(\varepsilon +P(\tau_0))e^{\frac12(\tau-\tau_0)}\bar{X}+(1+C\varepsilon )P(\tau_0)e^{(-\frac32+C\varepsilon)(\tau-\tau_0)}.
\end{align*}
Integrating the above inequality we obtain
\begin{align*}
e^{\frac12(\tau-\tau_0)}\bar{X}(\tau)&\leq \bar{X}(\tau_0)+C(\varepsilon +P(\tau_0))\int_{\tau_0}^\tau e^{\frac12(s-\tau_0)}\bar{X}(s)ds\\
&\quad+(1+C\varepsilon )P(\tau_0)\int_{\tau_0}^\tau e^{(-\frac32+C\varepsilon)(s-\tau_0)}ds\\
&\leq \bar{X}(\tau_0)+C(\varepsilon +P(\tau_0))\int_{\tau_0}^\tau e^{\frac12(s-\tau_0)}\bar{X}(s)ds
+\frac{(1+C\varepsilon)P(\tau_0)}{\frac32-C\varepsilon},
\end{align*}
which can be written for small $\varepsilon$ as follows:
\[
e^{\frac12(\tau-\tau_0)}\bar{X}(\tau)\leq \bar{X}(\tau_0)+P(\tau_0)+C(\varepsilon +P(\tau_0))\int_{\tau_0}^\tau e^{\frac12(s-\tau_0)}\bar{X}(s)ds.
\]
By Gronwall's inequality we obtain 
\[
e^{\frac12(\tau-\tau_0)}\bar{X}(\tau)\leq (\bar{X}(\tau_0)+P(\tau_0))e^{C(\varepsilon+P(\tau_0))(\tau-\tau_0)},
\]
which shows that
\[
\bar{X}(\tau)\leq (\bar{X}(\tau_0)+P(\tau_0))e^{(-\frac12+C(\varepsilon+P(\tau_0)))(\tau-\tau_0)}.
\]
Since $\bar{X}=fX$ with $f= 1+O(\varepsilon)$, we conclude that
\begin{align}\label{boot}
X(\tau)\leq (1+C\varepsilon)\Big(X(\tau_0)+P(\tau_0)\Big)e^{(-\frac12+C(\varepsilon+P(\tau_0)))(\tau-\tau_0)}.
\end{align}

\subsection{Closing the bootstrap argument and its result}

Our goal is to obtain an improved decay rate
\begin{align}
\Sigma_+\leq \varepsilon e^{(-1+\delta) (\tau-\tau_0)} \label{boot_Sigma_+}
\end{align}
on the interval $[\tau_0,\tau_1)$. If this is the case, there is nothing more to do. Let us suppose the opposite, namely that
for any $\delta>0$, there exist $\tau_-$ and $\tau_+$ such that $\tau_0\leq \tau_-<\tau_+\leq \tau_1$ and
\begin{align}
\vert \Sigma_+ \vert &\leq \varepsilon e^{(-1+\delta) (\tau-\tau_0)},\quad t\in [t_0,t_-],\nonumber\\
\vert \Sigma_+ \vert &\geq \varepsilon e^{(-1+\delta) (\tau-\tau_0)},\quad t\in [t_-,t_+).\label{op}
\end{align}
We consider the second time interval $[t_-,t_+)$.
From the evolution equation of $\Sigma_+$ we have
\begin{align*}
\frac{\Sigma_+'}{\Sigma_+}=-\frac32 +\frac{S}{6H^2}+X^2(\cos2\psi -1)\Sigma_+^{-1}+ X^2(\cos 2\psi+\frac12) +S_+\Sigma_+^{-1},
\end{align*}
where we assume that $\Sigma_+\neq 0$. From this we obtain
\begin{align*}
\frac{\Sigma_+'}{\Sigma_+}\leq -\frac32 +\frac{S}{6H^2}+CX^2 \vert \Sigma_+^{-1}\vert + C X^2 +\vert S_+\Sigma_+^{-1}\vert.
\end{align*}
Note that this inequality is invariant if we make the transformation $\Sigma_+ \rightarrow -\Sigma_+$. Using the estimates obtained and \eqref{op} we obtain
\begin{align*}
\frac{\Sigma_+'}{\Sigma_+} &\leq -\frac32 +  C(1+C\varepsilon)^2 \left[\Big(X(\tau_0)+P(\tau_0)\Big)^2+P(\tau_0)\right]e^{(-1+C(\varepsilon+P(\tau_0)))(\tau-\tau_0)} \varepsilon^{-1} e^{(1-\delta) (\tau-\tau_0)} \\
&\leq -1 +  C\varepsilon^{-1} (1+C\varepsilon)^2 \left[\Big(X(\tau_0)+P(\tau_0)\Big)^2+P(\tau_0)\right]e^{(-\delta+C(\varepsilon+P(\tau_0)))(\tau-\tau_0)}.
\end{align*}
Note that $X(\tau_0)$ and $P(\tau_0)$ can be chosen independently from $\varepsilon$ and can be made even smaller if necessary.
Integrating our last inequality we obtain a contradiction to \eqref{op}. Thus we have shown that there exists a small number $\delta>0$ such that $\Sigma_+$ satisfies \eqref{boot_Sigma_+}. We now combine the estimates \eqref{estimatem}, \eqref{boot}, and \eqref{boot_Sigma_+} to conclude that the estimates obtained hold globally in time provided that $M(\tau_0)$, $X(\tau_0)$, $\Sigma_+(\tau_0)$, and $P(\tau_0)$ are sufficiently small. We obtain the following results.

\begin{prop}
Consider any $C^{\infty}$ solution of the Einstein-Vlasov system with  with reflection and Bianchi VII$_0$ symmetries and with $C^{\infty}$ initial data. Assume that $P(\tau_0)$, ${\Sigma}_+(\tau_0)$, $X(\tau_0)$ and $M(\tau_0) $ are sufficiently small. Then at late times there exists a small constant  $\epsilon>0$ such the following estimates hold:
\begin{align*}
\Sigma_+&=O(\epsilon e^{(-1+\epsilon)\tau}),\\
X&=O(\epsilon e^{(-\frac12+\epsilon)\tau}),\\
M&=O(\epsilon e^{(-\frac12+\epsilon)\tau}),\\
\frac{S}{\rho}&=O(\epsilon e^{(-2+\epsilon)\tau}).
\end{align*}
\end{prop}

\section{Convergence of the metric, optimal estimates and the Weyl parameter}

\subsection{Convergence of the metric}
Let us now come back to the evolution equation of the Hubble variable \eqref{H-1} in terms of the new variables:
\begin{align*}
\partial_t(H^{-1})=\frac{3}{2}+\frac32 \Sigma_+^2 +\frac12 X^2+X^2 \cos 2\psi+\frac{S}{6H^2}.
\end{align*}
From the results of the previous section we have
\begin{align}\label{hybrid}
\partial_t(H^{-1})=\frac{3}{2}+O(\epsilon e^{(-1+2\epsilon)\tau}).
\end{align}
We choose $t_0=\frac23H^{-1}(t_0)$ and integrate the above to obtain $H^{-1}=(3/2) t (1+O(\epsilon))$ at late times. Since $dt/d\tau = H^{-1}$, we have $t'=(3/2) t (1+O(\epsilon))$ and obtain $t/t_0=\exp((3/2) (1+O(\epsilon))(\tau-\tau_0))$,
which means that
\[
e^{\tau-\tau_0}=\bigg(\frac{t}{t_0}\bigg)^{\frac23 +O(\epsilon)}.
\]
If we use this in \eqref{hybrid}, then we obtain
\[
\partial_t(H^{-1})=\frac32 +O(\epsilon t^{-\frac32 +O(\epsilon)}).
\]
This implies that
\begin{align*}
H=\frac23 t^{-1}(1+O(\epsilon t^{-\frac23+O(\epsilon)})),
\end{align*}
which one can compare with (3.34) of \cite{WHU}.
With the estimates of $\Sigma_+$ and $\Sigma_-$ we obtain 
\begin{align*}
\Sigma_+^2 + \Sigma_-^2= O(\epsilon t^{-\frac23+O(\epsilon)}),
\end{align*}
which implies that
\begin{align*}
\sigma^a_b \sigma^b_a= O(t^{-\frac83+O(\epsilon)}).
\end{align*}
We now observe that $(\sigma_{cd}\sigma^{cd})^\frac12$ is integrable.

\begin{prop}
Consider any $C^{\infty}$ solution of the Einstein-Vlasov system with reflection and Bianchi VII$_0$ symmetries and with $C^{\infty}$ initial data. Assume that $P(\tau_0)$, ${\Sigma}_+(\tau_0)$, $X(\tau_0)$ and $M(\tau_0)$ are sufficiently small and $t_0=(2/3)H^{-1}(t_0)>0$. Then at late times there exists a small constant  $\epsilon>0$ and constant matrices $\mathcal{G}_{ab}$ and $\mathcal{G}^{ab}$ such that the following estimates hold:
\begin{align}\label{metric1}
g_{ab}(t)&= t^{+\frac43}\Big(\mathcal{G}_{ab}+O( t^{-\frac13+\epsilon})\Big),\\
g^{ab}(t)&= t^{-\frac43}\Big(\mathcal{G}^{ab}+O( t^{-\frac13+\epsilon}) \Big).\label{metric2}
\end{align}
\end{prop}
\begin{proof}
Define $g_{ab}=t^{\frac43} \bar{g}_{ab}$.
From the definitions we have that
\begin{align*}
\dot{\bar{g}}_{ab}=2\left[\left(H-\frac23 t^{-1}\right)\bar{g}_{ab}-t^{-\frac43}\sigma_{ab}\right].
\end{align*}
Using the usual matrix norm we obtain that
\begin{align*}
\Vert \bar{g}_{ab}(t) \Vert \leq \Vert \bar{g}_{ab}(t_0) \Vert + C\int^t_{t_0}
\Big(\Big\vert H(s)- \frac23 s^{-1}\Big\vert +(\sigma_{cd}\sigma^{cd}(s))^{\frac{1}{2}}\Big) \Vert \bar{g}_{ab}(s) \Vert ds.
\end{align*}
With Gronwall's inequality we obtain
\begin{align*}
\Vert \bar{g}_{ab}(t) \Vert \leq \Vert \bar{g}_{ab}(t_0) \Vert \exp\bigg(
\int^t_{t_0} \Big\vert H(s)- \frac23 s^{-1}\Big\vert +(\sigma_{cd}\sigma^{cd}(s))^{\frac{1}{2}}ds\bigg)
\leq C\|\bar{g}_{ab}(t_0)\|,
\end{align*}
because what matters is that the integrand decays faster than $s^{-1}$.
Therefore, $\bar{g}_{ab}$ is bounded for all $t\geq t_0$, and we have
\begin{align*}
\vert t^{-\frac43}g_{ab} \vert \leq C.
\end{align*}
Moreover, we have $\|\sigma_{ab}\|\leq (\sigma_{cd}\sigma^{cd})^{\frac12}\|g_{ab}\|\leq Ct^{\epsilon}$ for some small constant $\epsilon>0$, and this is enough to conclude that there exists a constant matrix $\mathcal{G}_{ab}$ defined by
\[
\mathcal{G}_{ab}=\bar{g}_{ab}(t_0)+\int_{t_0}^\infty
2\Big(H(s)-\frac23 s^{-1}\Big)\bar{g}_{ab}(s)-2s^{-\frac{4}{3}}\sigma_{ab}(s)ds.
\]
Note that
\begin{align*}
|\mathcal{G}_{ab}-\bar{g}_{ab}(t)|
&\leq C \int_t^\infty s^{-\frac43+\epsilon}ds\leq C t^{-\frac13+\epsilon},
\end{align*}
and this proves the estimate \eqref{metric1}. The estimate \eqref{metric2} is obtained in a similar way, and this completes the proof.
\end{proof}

The result concerning the metric is the analogue of (C.3) of \cite{WHU} where one has to use (C.10) where $\beta$ is according  to (3.26) equal to $1/2$ in the dust case which corresponds to $\gamma=1$. 
We see thus that the metric becomes isotropic in the sense of the theorem. 

\subsection{Optimal estimates and the Weyl parameter}
With the estimates for the metric, we can get rid of the epsilons of the estimates obtained. For $P$ with the results obtained we get
\[
\frac{d}{d\tau}\Big[ e^{2\tau}P\Big]\leq 2e^{2\tau}(\Sigma^a_b\Sigma_a^b)^{\frac12}P\leq C\epsilon e^{(-\frac12+C\epsilon)\tau},
\]
for some small $\epsilon>0$. Integrating the above we obtain
\[
P(\tau)\leq C(\epsilon+P(\tau_0))e^{-2\tau}.
\]
Since $P(\tau_0)$ is assumed to be small, we can write $P=O(\epsilon e^{-2\tau})$. Similarly for $M$ we have
\begin{align*}
\frac{d}{d\tau}\Big[ e^{\frac12 \tau}M\Big]&=-e^{\frac12 \tau}M\bigg(\frac32 \Sigma_+^2 +\frac12 X^2 +\frac{S}{6H^2}+2\Sigma_+ +X^2(\cos 2\psi +3M \sin 2\psi)\bigg)\\
&\leq C\epsilon e^{(-\frac12 +C\epsilon)\tau},
\end{align*}
and this is enough to conclude that $M=O(\epsilon e^{-\frac12\tau})$. In a similar way we obtain optimal estimates for $X$ and $\Sigma_+$. 

As a consequence of the estimates, the Weyl parameter has the same late-time behaviour as in the case of a non-tilted perfect fluid, namely (3.39) of \cite{WHU}
\begin{align*}
\mathcal{W}= \frac{2X}{M} [1+O(M)].
\end{align*}
In particular we have shown that $X$ and $M$ have the same decay rate, so we obtain that the Weyl parameter tends to a constant and thus has the same behaviour as for dust. Let us summarise the results in the following theorem:
 \begin{thm}
Consider any $C^{\infty}$ solution of the Einstein-Vlasov system with reflection and Bianchi VII$_0$ symmetries and with $C^{\infty}$ initial data. Assume that $P(\tau_0)$, ${\Sigma}_+(\tau_0)$, $X(\tau_0)$ and $ M(\tau_0) $ are sufficiently small and $t_0=(2/3)H^{-1}(t_0)>0$. Then at late times there exists a small constant $\epsilon>0$ and constant matrices $\mathcal{G}_{ab}$ and $\mathcal{G}^{ab}$ such that the following estimates hold:
\begin{align*}
g_{ab}(t)&= t^{+\frac43}\Big(\mathcal{G}_{ab}+O( t^{-\frac13})\Big),\\
g^{ab}(t)&= t^{-\frac43}\Big(\mathcal{G}^{ab}+O( t^{-\frac13}) \Big).
\end{align*}
Moreover, we have
\begin{align*}
\Sigma_+&=O(e^{-\tau}),\\
X&=O(e^{-\frac12\tau}),\\
M&=O(e^{-\frac12\tau}),\\
\frac{S}{\rho}&=O(e^{-2\tau}),\\
H&=\frac23 t^{-1}(1+O(t^{-\frac23})),\\
\mathcal{W}&= C [1+O(e^{-\frac12\tau})].
\end{align*}
\end{thm}

\section{Conclusions and Outlook}

We have shown that under small data assumptions the space-time will not be self-similar at late times. In particular a dimensionless variable associated to the Weyl-curvature blows up. Nevertheless the shear tends to zero. This means that the methods which have been applied in \cite{EN} for solutions close to a self-similar solution, or what one might call an ``exact'' solution have also been applicable to a case where solutions are close to a non self-similar solution. The key to treat this case was to introduce the compactification of one of the variables as was done in \cite{WHU}. We have analysed the case of being close to dust. Using a perfect fluid for the dust case a Weyl-curvature bifurcation occurs and we have shown that for collisionless matter one obtains the limiting behaviour of the dust case. A natural generalisation is to consider the Boltzmann case along the methods of \cite{LN}. 
Near radiation another bifurcation occurs \cite{NHW}. We plan to analyse the massless Einstein-Vlasov case to see whether in the kinetic picture the conclusions remain the same. Here we have only treated the reflection symmetric case. It would thus be of interest to remove this restriction as was done in \cite{E4} concerning Bianchi II and VI$_0$. For this aim the variables introduced in \cite{HervikVII0} concerning a tilted perfect fluid will be useful and the results should be compared to that paper. 
\section*{Acknowledgements}
H. Lee has been supported by the TJ Park Science Fellowship of POSCO TJ Park Foundation. This research was supported by Basic Science Research Program through the National Research Foundation of Korea(NRF) funded by the Ministry of Science, ICT \& Future Planning (NRF-2015R1C1A1A01055216). E.N. is currently funded by a Juan de la Cierva research fellowship from the Spanish government and acknowledges financial support from the Spanish Ministry of Economy and Competitiveness, through the ''Severo Ochoa'' Programme for Centres of Excellence in R\& D (SEV-2015-0554). E.N. would also like to thank A.A. Coley, S. Hervik and C. Uggla for discussions about the subject of this paper.

\end{document}